\documentclass[11pt]{article}
\usepackage{amsmath,amsthm,amsfonts}
\DeclareMathOperator{\Hom}{Hom}
\DeclareMathOperator{\Com}{Com}

\DeclareMathOperator{\II}{I}
\DeclareMathOperator{\1}{id}

\newcommand{\NN}{\mathbb{N}}
\newcommand{\RR}{\mathbb{R}}

\newcommand{\EEnd}{\mathcal End}
\newcommand{\EE}{\mathcal E}
\newcommand{\bul}{\bullet}

\renewcommand{\=}{:=}
\renewcommand{\t}{\otimes}

\newcommand{\al}{\alpha}
\newcommand{\be}{\beta}
\newtheorem{thm}{Theorem}[section]
 \newtheorem{lemma}[thm]{Lemma}

\theoremstyle{definition}
 \newtheorem{defn}[thm]{Definition}
\theoremstyle{definition}
 \newtheorem{exam}[thm]{Example}
\theoremstyle{definition}

\numberwithin{equation}{section}
\begin{document}
\title{\LARGE\bf  2D binary operadic Lax representation\\ for harmonic oscillator}
\date{}
\author{\Large Eugen Paal and J\"{u}ri Virkepu\\ \\
Department of Mathematics, Tallinn University of Technology\\
Ehitajate tee 5, 19086 Tallinn, Estonia\\ \\
E-mails: eugen.paal@ttu.ee and jvirkepu@staff.ttu.ee}
\maketitle
\thispagestyle{empty}
\begin{abstract}
It is explained how the time evolution of operadic variables may be introduced by using the operadic Lax equation. As an example, a $2$-dimensional binary operadic Lax representation for the  harmonic oscillator  is 
constructed.
\end{abstract}

\section{Introduction}

In the Hamiltonian formalism, a mechanical system is described by canonical variables $q^i,p_i$ and their time evolution is prescribed by the Hamiltonian system
\begin{equation}
\label{ham}
\dfrac{dq^i}{dt}=\dfrac{\partial H}{\partial p_i},
\quad
\dfrac{dp_i}{dt}=-\dfrac{\partial H}{\partial q^i}
\end{equation}
By a Lax representation \cite{Lax68,BBT03} of a mechanical system one means such a pair $(L,M)$ of matrices (linear operators) $L,M$ that the above Hamiltonian system may be represented as the Lax equation
\begin{equation}
\label{lax}
\dfrac{dL}{dt}=[M,L]\= ML-LM
\end{equation}
Thus, from the algebraic point of view, mechanical systems can be described by linear operators, i.e by  linear maps $V\to V$  of a vector space $V$. As a generalization of this one can pose the following question \cite{Paal07}: 
how to describe the time evolution of the linear operations (multiplications) $V^{\t n}\to V$? 

The algebraic operations (multiplications) can be seen as an example of the \emph{operadic} variables \cite{Ger,GGS92,KP,KPS}. 
If an operadic system depends on time one can speak about \emph{operadic dynamics} \cite{Paal07}.
The latter may be introduced by simple and natural analogy with the Hamiltonian dynamics.
In particular, the time evolution of operadic variables may be given by the operadic Lax equation.
In \cite{PV07} it was shown how the dynamics may be introduced in a 2-dimensional Lie algebra. 
In the present paper, a $2$-dimensional binary operadic Lax representation for the  harmonic oscillator  is constructed. 

\section{Operad}

Let $K$ be a unital associative commutative ring, and let $C^n$ ($n\in\NN$) be unital $K$-modules. For $f\in C^n$, we refer to $n$ as the \emph{degree} of $f$ and often write (when it does not cause confusion) $f$ instead of $\deg f$. For example, $(-1)^f\= (-1)^n$, $C^f\= C^n$ and $\circ_f\= \circ_n$. Also, it is convenient to use the \emph{reduced} degree $|f|\= n-1$. Throughout this paper, we assume that $\t\= \t_K$.

\begin{defn}[operad (e.g \cite{Ger,GGS92})]
A linear (non-symmetric) \emph{operad} with coefficients in $K$ is a sequence $C\= \{C^n\}_{n\in\NN}$ of unital
$K$-modules (an $\NN$-graded $K$-module), such that the following conditions hold:
\begin{enumerate}
\item[(1)]
For $0\leq i\leq m-1$ there exist the \emph{partial compositions}
\[
  \circ_i\in\Hom(C^m\t C^n,C^{m+n-1}),\quad |\circ_i|=0
\]
\item[(2)]
For all $h\t f\t g\in C^h\t C^f\t C^g$,
the \emph{composition (associativity) relations} hold,
\[
(h\circ_i f)\circ_j g=
\begin{cases}
    (-1)^{|f||g|} (h\circ_j g)\circ_{i+|g|}f
                       &\text{if $0\leq j\leq i-1$},\\
    h\circ_i(f\circ_{j-i}g)  &\text{if $i\leq j\leq i+|f|$},\\
    (-1)^{|f||g|}(h\circ_{j-|f|}g)\circ_i f
                       &\text{if $i+f\leq j\leq|h|+|f|$}.
\end{cases}
\]
\item[(3)]
A unit $\II\in C^1$ exists such that
\[
\II\circ_0 f=f=f\circ_i \II,\quad 0\leq i\leq |f|
\]
\end{enumerate}
\end{defn}

In the second item, the \emph{first} and \emph{third} parts of the
defining relations turn out to be equivalent.

\begin{exam}[endomorphism operad \cite{Ger}]
\label{HG} Let $V$ be a unital $K$-module and
$\EE_V^n\= {\EEnd}_V^n\= \Hom(V^{\t n},V)$. Define the partial compositions
for $f\t g\in\EE_V^f\t\EE_V^g$ as
\[
f\circ_i g\= (-1)^{i|g|}f\circ(\1_V^{\t i}\t g\t\1_V^{\t(|f|-i)}),
         \quad 0\leq i\leq |f|
\]
Then $\EE_V\= \{\EE_V^n\}_{n\in\NN}$ is an operad (with the unit $\1_V\in\EE_V^1$) called the
\emph{endomorphism operad} of $V$.
\end{exam}
Therefore, algebraic operations can be seen as elements of the endomorphism operad.
Just as elements of a vector space are called \emph{vectors},  it is natural to call elements of an abstract operad \emph{operations}. 

\section{Gerstenhaber brackets and operadic Lax equation}

\begin{defn}[total composition \cite{Ger,GGS92}]
The \emph{total composition} $\bul\:C^f\t C^g\to C^{f+|g|}$ is defined in an operad $C$  by
\[
f\bul g\= \sum_{i=0}^{|f|}f\circ_i g \quad \in C^{f+|g|},
\quad |\bul|=0
\]
The pair $\Com C\= \{C,\bul\}$ is called the \emph{composition algebra} of $C$.
\end{defn}

\begin{defn}[Gerstenhaber brackets \cite{Ger,GGS92}]
The  \emph{Gerstenhaber brackets} $[\cdot,\cdot]$ are defined in $\Com C$ as a graded commutator by
\[
[f,g]\= f\bul g-(-1)^{|f||g|}g\bul f=-(-1)^{|f||g|}[g,f],\quad|[\cdot,\cdot]|=0
\]
\end{defn}

The \emph{commutator algebra} of $\Com C$ is denoted as $\Com^{-}\!C\= \{C,[\cdot,\cdot]\}$.
One can prove that $\Com^-\!C$ is a \emph{graded Lie algebra}. The Jacobi
identity reads
\[
(-1)^{|f||h|}[[f,g],h]+(-1)^{|g||f|}[[g,h],f]+(-1)^{|h||g|}[[h,f],g]=0
\]

Assume that $K\= \RR$ and operations are differentiable.
Dynamics in operadic systems (operadic dynamics) may be introduced by

\begin{defn}[operadic Lax pair \cite{Paal07}]
Allow a classical dynamical system to be described by the Hamiltonian system (\ref{ham}).
An \emph{operadic Lax pair} is a pair $(L,M)$ of the homogeneous operations $L,M\in C$,
such that the Hamiltonian system  (\ref{ham}) may represented as the \emph{operadic Lax equation}
\[
\dfrac{dL}{dt}=[M,L]\= M\bul L-(-1)^{|M||L|}L\bul M
\]
Evidently, the degree constraints $|L|=|M|=0$ give rise to ordinary Lax equation (\ref{lax}) \cite{Lax68,BBT03}.
\end{defn}

\section{Operadic harmonic oscillator}

Consider the Lax pair for the harmonic oscillator:
\[
L=\begin{pmatrix}
p&\omega q\\
\omega q &-p
\end{pmatrix},
\quad
M=\frac{\omega}{2}
\begin{pmatrix}
0&-1\\
1&0
\end{pmatrix}
\]
Since the Hamiltonian is
\[
H(q,p)=\frac{1}{2}(p^2+\omega^2q^2)
\]
it is easy to check that the Lax equation
\[
\dot{L}=[M,L]\=  ML - LM
\]
is equivalent to the Hamiltonian system
\begin{equation}
\label{ham_osc}
\dfrac{dq}{dt}=\dfrac{\partial H}{\partial p}=p,
\quad
\dfrac{dp}{dt}=-\dfrac{\partial H}{\partial q}=-\omega^2q
\end{equation}
If $\mu$ is a homogeneous operadic variable one can use the above Hamilton equations to obtain
\[
\dfrac{d\mu}{dt}
=\dfrac{\partial\mu}{\partial q}\dfrac{dq}{dt}+\dfrac{\partial\mu}{\partial p}\dfrac{dp}{dt}
=p\dfrac{\partial\mu}{\partial q}-\omega^2q\dfrac{\partial\mu}{\partial p}
=[M,\mu]
\]
Therefore, we get the following linear partial differential equation for the operadic variable $\mu(q,p)$:
\[
p\dfrac{\partial\mu}{\partial q}-\omega^2q\dfrac{\partial\mu}{\partial p}=M\bul\mu- \mu\bul M
\]
By integrating one gains sequences of operations  called the \emph{operadic (Lax representations for) harmonic oscillator}.

\section{Main example}

Let $A\= \{V,\mu\}$ be a  binary algebra with an operation $xy\= \mu(x\t y)$.
We require that $\mu=\mu(q,p)$ so that $(\mu,M)$ is an operadic Lax pair, i.e the operadic Lax equation
\[
\dot{\mu}=[M,\mu]\=  M\bul\mu-\mu\bul M,\quad |\mu|=1,\quad |M|=0
\]
is equivalent to the Hamiltonian system of the harmonic oscillator.

Let $x,y\in V$. Assuming that $|M|=0$ and $|\mu|=1$, one has
\begin{align*}
M\bul\mu
&=\sum_{i=0}^0(-1)^{i|\mu|}M\circ_i\mu
=M\circ_0\mu=M\circ\mu\\
\mu\bul M &=\sum_{i=0}^1(-1)^{i|M|}\mu\circ_i M =\mu\circ_0
M+\mu\circ_1 M=\mu\circ(M\t\1_V)+\mu\circ(\1_V\t M)
\end{align*}
Therefore, one has
\[
\dfrac{d}{dt}(xy)=M(xy)-(Mx)y-x(My)
\]
Let $\dim V=n$.
In a basis $\{e_1,\ldots,e_n\}$ of $V$,  the structure constants $\mu_{jk}^i$ of $A$ are defined by
\[
\mu(e_j\t e_k)\=  \mu_{jk}^i e_i,\quad j,k=1,\ldots,n
\]
In particular,
\[
\dfrac{d}{dt}(e_je_k)=M(e_je_k)-(Me_j)e_k-e_j(Me_k)
\]
By denoting $Me_i\=  M_i^se_s$, it follows that
\[
\dot{\mu}_{jk}^i=\mu_{jk}^sM_s^i-M_j^s\mu_{sk}^i-M_k^s\mu_{js}^i,\quad i,j,k=1,\ldots, n
\]
In particular, one has

\begin{lemma}
\label{lemma:first}
Let $\dim V=2$ and
$
M\= (M_j^i)\= 
\frac{\omega}{2}
\left(
\begin{smallmatrix}
0&-1\\
1&0
\end{smallmatrix}
\right)
$.
Then the $2$-dimensional binary operadic Lax equations read
\[
\begin{cases}
\dot{\mu}_{11}^{1}=-\frac{\omega}{2}\left(\mu_{11}^{2}+\mu_{12}^{1}+\mu_{21}^{1}\right),\quad
\dot{\mu}_{11}^{2}=\frac{\omega}{2}\left(\mu_{11}^{1}-\mu_{12}^{2}-\mu_{21}^{2}\right)\\
\dot{\mu}_{12}^{1}=-\frac{\omega}{2}\left(\mu_{12}^{2}-\mu_{11}^{1}+\mu_{22}^{1}\right),\quad
\dot{\mu}_{12}^{2}=\frac{\omega}{2}\left(\mu_{12}^{1}+\mu_{11}^{2}-\mu_{22}^{2}\right)\\
\dot{\mu}_{21}^{1}=-\frac{\omega}{2}\left(\mu_{21}^{2}-\mu_{11}^{1}+\mu_{22}^{1}\right),\quad
\dot{\mu}_{21}^{2}=\frac{\omega}{2}\left(\mu_{21}^{1}+\mu_{11}^{2}-\mu_{22}^{2}\right)\\
\dot{\mu}_{22}^{1}=-\frac{\omega}{2}\left(\mu_{22}^{2}-\mu_{12}^{1}-\mu_{21}^{1}\right),\quad
\dot{\mu}_{22}^{2}=\frac{\omega}{2}\left(\mu_{22}^{1}+\mu_{12}^{2}+\mu_{21}^{2}\right)\\
\end{cases}
\]
\end{lemma}

For the harmonic oscillator, define its auxiliary functions $A_\pm$ and $D_\pm$ by
\begin{equation}
\label{def_A}
\begin{cases}
A_+^2+A_-^2=2\sqrt{2H}\\
A_+^2-A_-^2=2p\\
A_+A_-=\omega q\\
\end{cases},\quad
\begin{cases}
D_+\= \frac{A_+}{2}(A_+^2-3A_-^2)\\
D_-\= \frac{A_-}{2}(3A_+^2-A_-^2)\\
\end{cases}
\end{equation}
Differentiating the defining relations (\ref{def_A}) of $A_{\pm}$  with respect to $t$ one gets
 \begin{equation}
\label{diff_A}
\begin{cases}
A_+\dot{A}_++A_-\dot{A}_-=\frac{1}{\sqrt{2H}}(p\dot{p}+\omega^{2}q\dot{q})\\
A_+\dot{A}_+-A_-\dot{A}_-=\dot{p}\\
A_-\dot{A}_++A_+\dot{A}_-=\omega \dot{q}
\end{cases}
\end{equation}
Now one can propose
\begin{thm}
Let $C_{\nu}\in\mathbb{R}$ ($\nu=1,\ldots,8$) be arbitrary real--valued parameters,
$
M\= 
\frac{\omega}{2}
\left(
\begin{smallmatrix}
0&-1\\
1&0
\end{smallmatrix}
\right)
$
and
\[
\begin{cases}
\mu_{11}^{1}(q,p)=\hphantom{-}C_5A_-+C_6A_++C_7D_-+C_8D_+\\
\mu_{12}^{1}(q,p)=\hphantom{-}C_1A_++C_2A_--C_7D_++C_8D_-\\
\mu_{21}^{1}(q,p)=-C_1A_+-C_2A_--C_3A_+-C_4A_--C_5A_++C_6A_--C_7D_++C_8D_-\\
\mu_{22}^{1}(q,p)=-C_3A_-+C_4A_+-C_7D_--C_8D_+\\
\mu_{11}^{2}(q,p)=\hphantom{-}C_3A_++C_4A_--C_7D_++C_8D_-\\
\mu_{12}^{2}(q,p)=\hphantom{-}C_1A_--C_2A_++C_3A_--C_4A_++C_5A_-+C_6A_+-C_7D_--C_8D_+\\
\mu_{21}^{2}(q,p)=-C_1A_-+C_2A_+-C_7D_--C_8D_+\\
\mu_{22}^{2}(q,p)=-C_5A_++C_6A_-+C_7D_+-C_8D_-\\
\end{cases}
\]
Then $(\mu,M)$ is a $2$-dimensional binary operadic Lax pair of the harmonic oscillator.
\end{thm}

\begin{proof}[Proof]
Denote
\[
\begin{cases}
G_{\pm}^{\omega/2}&\= \dot{A}_{\pm}\pm\frac{\omega}{2}A_{\mp}\\
G_{\pm}^{3\omega/2}&\= \dot{D}_{\pm}\pm\frac{3\omega}{2}D_{\mp}\\
\end{cases}
\]
Define the matrix
\[
\Gamma
=(\Gamma_{\al}^{\be})\=\begin{pmatrix}
                                   0 & \hphantom{-}G_+^{\omega/2} & -G_+^{\omega/2} & 0 & 0 & \hphantom{-}G_-^{\omega/2} & -G_-^{\omega/2} & 0 \\
                                   0 & \hphantom{-}G_-^{\omega/2} & -G_-^{\omega/2} & 0 & 0 & -G_+^{\omega/2} & \hphantom{-}G_+^{\omega/2} & 0 \\
                                   0 & 0 & -G_+^{\omega/2} & -G_-^{\omega/2} & \hphantom{-}G_+^{\omega/2} & \hphantom{-}G_-^{\omega/2} & 0 & 0 \\
                                   0 & 0 & -G_-^{\omega/2} & \hphantom{-}G_+^{\omega/2} & \hphantom{-}G_-^{\omega/2} & -G_+^{\omega/2} & 0 & 0 \\
                                   G_-^{\omega/2} & 0 & -G_+^{\omega/2} & 0 & 0 & \hphantom{-}G_-^{\omega/2} & 0 & -G_+^{\omega/2} \\
                                   G_+^{\omega/2} & 0 & \hphantom{-}G_-^{\omega/2} & 0 & 0 & \hphantom{-}G_+^{\omega/2} & 0 & \hphantom{-}G_-^{\omega/2} \\
                                   G_-^{3\omega/2} & -G_+^{3\omega/2} & -G_+^{3\omega/2} & -G_-^{3\omega/2} & -G_+^{3\omega/2} & -G_-^{3\omega/2} & -G_-^{3\omega/2} & \hphantom{-}G_+^{3\omega/2} \\
                                   G_+^{3\omega/2} & \hphantom{-}G_-^{3\omega/2} & \hphantom{-}G_-^{3\omega/2} & -G_+^{3\omega/2} & \hphantom{-}G_-^{3\omega/2} & -G_+^{3\omega/2} & -G_+^{3\omega/2} & -G_-^{3\omega/2} \\
                                 \end{pmatrix}
\]
Then, by using Lemma \ref{lemma:first}, it follows that the $2$-dimensional binary operadic Lax equations read
\[
C_{\be}\Gamma_{\al}^{\be}=0,\quad \al=1,\ldots,8
\]
Since parameters $C_\be$ are arbitrary, the latter constraints imply  $\Gamma=0$.
Thus one has to consider the differential equations
\[
G_{\pm}^{\omega/2}=0=G_{\pm}^{3\omega/2}
\]
We show that
\[
\begin{cases}
\dot{p}=-\omega^{2}q\\
\dot{q}=p\\
\end{cases}\quad
\stackrel{(I)}{\Longleftrightarrow}\quad
 G_{\pm}^{\omega/2}=0 \quad
\stackrel{(II)}{\Longleftrightarrow}\quad G_{\pm}^{3\omega/2}=0
\]

First prove ($I$). 
$\Longrightarrow:$ Assume that the Hamilton equations (\ref{ham_osc}) for the harmonic oscillator hold. Then
it follows from (\ref{diff_A}) that
\begin{align*}
\begin{cases}
A_+\dot{A}_++A_-\dot{A}_-=0\\
A_+\dot{A}_+-A_-\dot{A}_-=-\omega^{2}q\\
A_-\dot{A}_++A_+\dot{A}_-=\omega p
\end{cases}
\quad&\Longleftrightarrow\quad
\begin{cases}
2A_-\dot{A}_-=\hphantom{-}\omega^{2}q\\
2A_+\dot{A}_+=-\omega^{2}q\\
A_-\dot{A}_++A_+\dot{A}_-=\omega p
\end{cases}\\
\quad&\Longleftrightarrow\quad
\begin{cases}
\dot{A}_-=\frac{\omega^{2}q}{2A_-}=\frac{\omega^{2}qA_+}{2A_-A_+}=\hphantom{-}\frac{\omega}{2}A_+\\
\dot{A}_+=\frac{-\omega^{2}q}{2A_+}=\frac{-\omega^{2}qA_-}{2A_+A_-}=-\frac{\omega}{2}A_-\\
A_+^{2}-A_-^{2}=2p
\end{cases}\\
\quad&\Longleftrightarrow\quad G_{\pm}^{\omega/2}=0
\end{align*}
and the latter the required system for $A_{\pm}$.
\newline
$\Longleftarrow:$ Assume that the differential equations $G_{\pm}^{\omega/2}=0$ hold. Then it follows from (\ref{diff_A}) that
\begin{align*}
\begin{cases}
A_-A_+-A_+A_-=\frac{2(p\dot{p}+\omega^{2}q\dot{q})}{\omega\sqrt{2H}}\\
A_+A_-+A_-A_+=-\frac{2}{\omega}\dot{p}\\
A_+^{2}-A_-^{2}=2\dot{q}
\end{cases}
\quad&\Longleftrightarrow\quad
\begin{cases}
p\dot{p}+\omega^{2}q\dot{q}=0\\
A_+A_-=-\frac{1}{\omega}\dot{p}\\
A_+^{2}-A_-^{2}=2\dot{q}
\end{cases}\\
\quad&\Longleftrightarrow\quad
\begin{cases}
p\dot{p}+\omega^{2}q\dot{q}=0\\
\dot{p}=-\omega A_+A_-=-\omega^{2}q\\
\dot{q}=\frac{1}{2}(A_+^{2}-A_-^{2})=p\\
\end{cases}
\end{align*}
where the first relation easily follows from the Hamiltonian system (\ref{ham_osc}).

Now prove ($II$).  Differentiate the auxiliary functions $D_{\pm}$ to get
\[
\begin{cases}
\dot{D}_+=\frac{1}{2}\dot{A}_+(A_+^{2}-3A_-^{2})+A_+(A_+\dot{A}_+-3A_-\dot{A}_-)\\
\dot{D}_-=\frac{1}{2}\dot{A}_-(3A_+^{2}-A_-^{2})+A_-(3A_+\dot{A}_+-A_-\dot{A}_-)
\end{cases}
\]
$\Longrightarrow:$ Assume that functions $A_{\pm}$ satisfy the differential equations $G_{\pm}^{\omega/2}=0$. Then
\[
\begin{cases}
\dot{D}_+=-\frac{\omega}{4}A_-(A_+^{2}-3A_-^{2})-\frac{A_+\omega}{2}(A_+A_-+3A_-A_+)\\
\dot{D}_-=\hphantom{-}\frac{\omega}{4}A_+(3A_+^{2}-A_-^{2})-\frac{A_-\omega}{2}(3A_+A_-+A_-A_+)
\end{cases}
\]
and
\[
\begin{cases}
\dot{D}_+=-\frac{3\omega}{2}\frac{A_-}{2}(3A_+^{2}-A_-^{2})=-\frac{3\omega}{2}D_-\\
\dot{D}_-=\hphantom{-}\frac{3\omega}{2}\frac{A_+}{2}(A_+^{2}-3A_-^{2})=\hphantom{-}\frac{3\omega}{2}D_+
\end{cases}
\quad\Longleftrightarrow\quad G_{\pm}^{3\omega/2}=0
\]
$\Longleftarrow:$ Assume that functions $D_{\pm}$ satisfy the differential equations $G_{\pm}^{3\omega/2}=0$. Then
\begin{align*}
&\begin{cases}
-\frac{3\omega}{2}D_-=\frac{\dot{A}_+}{2}(A_+^{2}-3A_-^{2})+A_+(A_+\dot{A}_+-3A_-\dot{A}_-)\\
\hphantom{-}\frac{3\omega}{2}D_+=\frac{\dot{A}_-}{2}(3A_+^{2}-A_-^{2})+A_-(3A_+\dot{A}_+-A_-\dot{A}_-)
\end{cases}\\
&\Longleftrightarrow
\begin{cases}
\dot{A}_+(3A_+^{2}-3A_-^{2})+\dot{A}_-(-6A_-A_+)=-3\omega D_-\\
\dot{A}_+(6A_+A_-)+\dot{A}_-(3A_+^{2}-3A_-^{2})=3\omega D_+
\end{cases}
\\
&\Longleftrightarrow
\begin{cases}
\hphantom{\omega}p\dot{A}_+-\omega q\dot{A}_-=-\frac{\omega}{2}D_-\\
\omega
q\dot{A}_++\hphantom{\omega}p\dot{A}_-=\hphantom{-}\frac{\omega}{2}D_+
\end{cases}
\end{align*}
To use the Cramer formulae, calculate
\begin{align*}
\Delta &=\begin{vmatrix}
           p & -\omega q \\
           \omega q & p \\
         \end{vmatrix}
=p^{2}+\omega^{2}q^{2}=2H\\
\Delta_{\dot{A}_+}&=\begin{vmatrix}
                     -\frac{\omega}{2}D_- & -\omega q \\
                     \hphantom{-}\frac{\omega}{2}D_+ & p \\
                   \end{vmatrix}
=-\frac{\omega}{2}(D_-p-D_+\omega q)\\
\Delta_{\dot{A}_-}&=\begin{vmatrix}
                     p & -\frac{\omega}{2}D_- \\
                     \omega q & \hphantom{-}\frac{\omega}{2}D_+ \\
                   \end{vmatrix}
=\frac{\omega}{2}(D_+p+D_-\omega q)
\end{align*}
Note that
\begin{align*}
D_-p-D_+\omega
q&=\frac{A_-}{2}p(3A_+^{2}-A_-^{2})-\frac{A_+}{2}\omega
q(A_+^{2}-3A_-^{2})\\
&=\frac{A_-}{2}\frac{1}{2}(A_+^{2}-A_-^{2})(3A_+^{2}-A_-^{2})-\frac{A_+}{2}A_+A_-(A_+^{2}-3A_-^{2})\\
&=\frac{A_-}{4}(A_+^{2}+A_-^{2})^{2}=2A_-H\\
D_+p+D_-\omega
q&=\frac{A_+}{2}p(A_+^{2}-3A_-^{2})+\frac{A_-}{2}\omega
q(3A_+^{2}-A_-^{2})\\
&=\frac{A_+}{2}\frac{1}{2}(A_+^{2}-A_-^{2})(A_+^{2}-3A_-^{2})+\frac{A_-}{2}A_+A_-(3A_+^{2}-A_-^{2})\\
&=\frac{A_+}{4}(A_+^{2}+A_-^{2})^{2}=2A_+H\\
\end{align*}
Thus,
\[
\begin{cases}
\dot{A}_+=\frac{\Delta_{\dot{A}_+}}{\Delta}=-\frac{\omega}{2}\frac{2HA_-}{2H}=-\frac{\omega}{2}A_-\\
\dot{A}_-=\frac{\Delta_{\dot{A}_-}}{\Delta}=\hphantom{-}\frac{\omega}{2}\frac{2HA_+}{2H}=\hphantom{-}\frac{\omega}{2}A_+
\end{cases}
\quad\Longleftrightarrow\quad G_{\pm}^{\omega/2}=0 \qedhere
\]
\end{proof}

\section*{Acknowledgement}
Research was in part supported by the Estonian Science Foundation, Grant ETF 6912.


\begin{thebibliography}{99}
\itemsep-3pt

\bibitem{BBT03}
\newblock{O.~Babelon, D.~Bernard, and M.~Talon.}
\newblock{Introduction to Classical Integrable Systems.}
\newblock{Cambridge Univ. Press, 2003.}

\bibitem{Ger}
\newblock{M.~Gerstenhaber.}
\newblock{The cohomology structure of an associative ring.}
\newblock{Ann. of Math. {\bf 78} (1963), 267-288.}

\bibitem{GGS92}
\newblock{M.~Gerstenhaber, A.~Giaquinto, and S.~D.~Schack.}
\newblock{Algebras, bialgebras, quantum groups, and algebraic deformations.}
\newblock{Contemp. Math. {\bf 134} (1992), 51-92.}

\bibitem{KP}
\newblock{L. Kluge and E. Paal.}
\newblock{On derivation deviations in an abstract pre-operad.}
\newblock{Comm. Algebra, {\bf29} (2001), 1609-1626.}

\bibitem{KPS}
\newblock{L. Kluge, E. Paal, and J. Stasheff.}
\newblock{Invitation to composition.}
\newblock{Comm. Algebra, {\bf28} (2000), 1405-1422.}

\bibitem{Lax68}
\newblock{P.~D.~ Lax.}
\newblock{Integrals of nonlinear equations of evolution and solitary waves.}
\newblock{Comm. Pure Applied Math. {\bf21} (1968), 467-490.}

\bibitem{Paal07}
\newblock{E.~Paal.}
\newblock{Invitation to operadic dynamics.}
\newblock{J. Gen. Lie Theory Appl. {\bf1} (2007), 57-63.}

\bibitem{PV07}
\newblock{E.~Paal and J.~Virkepu.}
\newblock{Note on operadic harmonic oscillator.}
\newblock{Rep. Math. Phys. {\bf 61} (2008), 205-210.}

\end{thebibliography}
\end{document}